\documentclass[12pt, final]{l4dc2023}

\usepackage{times}
\usepackage{graphicx}
\usepackage{epstopdf}
\usepackage{wrapfig}

\usepackage{times}
\usepackage{amsfonts}       % blackboard math symbols
\usepackage{amsmath} % assumes amsmath package installed
\usepackage{algorithm}
\usepackage{algpseudocode}

\usepackage{algorithm2e}[ruled]

\usepackage{graphicx} 
\usepackage{tcolorbox}
\usepackage{enumitem}
\usepackage{sidecap}
\usepackage{caption}
\usepackage{mathtools}
\usepackage[export]{adjustbox}

\usepackage{enumitem}
\usepackage{multirow}
\usepackage{makecell}

\author{%
 \Name{Taha Entesari} \Email{tentesa1@jhu.edu}\\
 \addr \normalfont Department of Electrical and Computer Engineering, Johns Hopkins University, USA
 \AND
 \Name{Mahyar Fazlyab} \Email{mahyarfazlyab@jhu.edu}\\
 \addr \normalfont Department of Electrical and Computer Engineering, Johns Hopkins University, USA
}

\begin{document}

\title[Automated Reachability Analysis of Neural Network-Controlled Systems]{Automated Reachability Analysis of Neural Network-Controlled Systems via Adaptive Polytopes}
\maketitle

\begin{abstract}%
Over-approximating the reachable sets of dynamical systems is a fundamental problem in safety verification and robust control synthesis. The representation of these sets is a key factor that affects the computational complexity and the approximation error.  In this paper, we develop a new approach for over-approximating the reachable sets of neural network dynamical systems using adaptive template polytopes.  We use the singular value decomposition of linear layers along with the shape of the activation functions to adapt the geometry of the polytopes at each time step to the geometry of the true reachable sets. 
We then propose a branch-and-bound method to compute accurate over-approximations of the reachable sets by the inferred templates. We illustrate the utility of the proposed approach in the reachability analysis of linear systems driven by neural network controllers.
\end{abstract}


\section{Introduction}

As the use of neural networks has expanded into safety-critical applications such as autonomous systems and automated healthcare, there is a growing need to develop efficient and scalable methods to rigorously verify neural networks against input uncertainties.  The canonical problem is to verify that for a bounded set of inputs, the reachable set of a trained model does not intersect with an unsafe set. From an optimization perspective, this problem can be formulated as a constraint satisfaction feasibility problem, where the goal is to either verify the constraint or find a counter-example.
%a
%However, solving these feasibility problems efficiently and at scale is challenging due to the large-scale and complex nature of neural networks.

%Computing the reachable sets of neural networks, however,
%the need for formal verification of proper functionality under real-world and adversarial scenarios is paramount. While such systems exhibit sound performance under nominal conditions, lack of scalable tools for providing formal safety guarantees has restricted their deployment in applications  requiring real-time verification \cite{rober2022hybrid}. Verifying the sound performance of large-scale neural networks in such systems has proven to be challenging.

%Verification of closed-loop systems with neural networks in closed-loop settings is fundamentally more difficult as they require the computation of the reachable sets. 

Going beyond machine learning, neural networks
also arise as function approximators in feedback control. Compared to open-loop settings, verification of neural networks in closed-loop systems is a more challenging problem, as it requires an explicit characterization of the reachable sets at each iteration.
Computation of reachable sets for dynamical systems is a fundamental problem that arises in, for example, safety verification and robust control synthesis--see \cite{althoff2021set} for an overview.
Formally, given the dynamical system $x^{k+1} = F(x^k) \quad x^0 \in \mathcal{X}^0$,
%
% \begin{align}
%     x^{k+1} = F(x^k) \quad x^0 \in \mathcal{X}^0
% \end{align}
where $\mathcal{X}^0$ is a bounded set of initial conditions, the reachable set at time $k+1$ is defined as 
\begin{align}\label{eqn:reachabeSetDefinition}
    \mathcal{X}^{k+1} = F(\mathcal{X}^k) = \{F(x) \mid x \in \mathcal{X}^k\},
\end{align}
Since it is generally difficult to compute the reachable sets exactly, these sets are often over-approximated iteratively by a sequence of \emph{template sets} $(\bar{\mathcal{X}}_k)_{k\geq 0}$: starting with  $\bar{\mathcal{X}}^0 = \mathcal{X}^0$, we over-approximate the image of $\bar{\mathcal{X}^k}$ under $F$ using a \emph{set propagation algorithm} to ensure that $F(\bar {\mathcal{X}}^k) \subseteq \bar {\mathcal{X}}^{k+1}$ for $k=0,1,\cdots$. If the over-approximated sets do not intersect with a set of unsafe states (e.g., obstacles), then safety can be guaranteed. However, the  over-approximation error can quickly accumulate over time (known as the wrapping effect \cite{neumaier1993wrapping}), leading to overly conservative bounds for long time horizons. 
%Hence, a notorious challenge in reachability analysis is to mitigate the accumulation of the over-approximation error. 
This challenge becomes even more pronounced when neural networks are involved in the feedback loop (e.g., when $F$ is itself a neural network approximation of an ODE).%, which is the subject of this paper. 

The choice of template sets can have a crucial effect on the accuracy of computations. Indeed, any potential mismatch between the shape of the template set and the actual reachable set can lead to conservative bounds (shape mismatch error). 
To minimize this error, it is essential to use \emph{dynamic} template sets that can adapt to the geometry of the reachable sets based on the structure of $F$. 
Furthermore, the method by which we propagate the sets through $F$ can incur conservatism due to the underlying relaxations (propagation error). 
This error can be mitigated by using less conservative relaxations in the propagation method and/or by partitioning the input set.

\paragraph{Our Contributions} In this paper, we propose a novel method for reachability analysis of  discrete-time affine systems in feedback with $\mathrm{ReLU}$\footnote{Rectified Linear Unit} neural network controllers. Using bounded polyhedra to represent the template sets, we propose a method that dynamically adapts the geometry of the template sets to that of reachable sets based on the structure of the closed-loop map. Thus, this approach eliminates the manual selection of the template directions, making the procedure fully automated. Based on the chosen template directions, we then compute tight\footnote{We globally solve the corresponding non-convex problems within an arbitrary accuracy.} polyhedral over-approximations of the reachable sets using efficient branch-and-bound (BnB) algorithms. Our method is modular in that it can incorporate any bound propagation method for neural networks. Our code is available at \url{https://github.com/o4lc/AutomatedReach.git}.

\subsection{Related Work}
\paragraph{Open-loop Verification} 
Verifying piecewise linear networks can be cast as an MILP with the binary variables describing the $\mathrm{ReLU}$ neurons \cite{cheng2017maximum,tjeng2017evaluating,dutta2018output,lomuscio2017approach,fischetti2018deep}. These problems can be solved globally using generic BnB methods, in which the optimization problem is recursively divided into sub-problems by branching the binary variables. The optimal value of each sub-problem is then bounded using convex (linear) relaxations, leading to provable bounds on the optimal objective value of the original problem. However, generic BnB solvers do not exploit the underlying structure of the problem and hence, may be inefficient. As such, state-of-the-art methods for complete  verification develop customized BnB methods in which efficient bound propagation methods are used \cite{wang2021beta, bunel2020lagrangian, de2021improved, xu2020fast, kouvaros2021towards, ferrari2022complete}. 
%
%State-of-the-art verification methods for neural networks in isolation use customized BnB frameworks to form complete verifiers. BnB methods work by partitioning the problem into smaller problems and solving any convex relaxation of the problem on each segment \cite{wang2021beta, bunel2020lagrangian, de2021improved, xu2020fast, kouvaros2021towards, ferrari2022complete}. 
% Any convex relaxation \cite{fazlyab2020safety, wong2018provable, singh2019abstract} can be used to form exact solvers. 
Branching is done either on binary variables that describe the $\mathrm{ReLU}$ activations or on the input set.
% Most works branch in the input space or on the $\mathrm{ReLU}$ activations. 
\cite{vincent2021reachable} branches using the activation pattern of the $\mathrm{ReLU}$ nonlinearities and uses polyhedral sections in the input space, making the neural network an affine function in each polyhedron. 

% \paragraph{Open-loop verification.} A wide variety of methods have been developed that aim at solving relaxations of problem \eqref{eqn:propertyVerification}. Some works use semidefinite programming \cite{raghunathan2018semidefinite, batten2021efficient, fazlyab2020safety}, others employ linear relaxations \cite{wong2018provable, singh2019abstract}. For piecewise linear neural networks, there also exist some complete methods that solve for the exact value \eqref{eqn:propertyVerification} using mixed integer programming \cite{tjeng2017evaluating}. Other methods employ BnB frameworks to form complete verifiers \cite{wang2021beta, bunel2020lagrangian, de2021improved, xu2020fast, kouvaros2021towards, ferrari2022complete}. BnB methods for piecewise linear networks typically branch the input space or the $\mathrm{ReLU}$ neurons. \cite{vincent2021reachable} branches using the activation pattern of the $\mathrm{ReLU}$ nonlinearities and uses polyhedral sections in the input space making the neural network an affine function in each polyhedron. 

\paragraph{Closed-loop Verification}
When the choice of set representation is template polyhedra, the problem of reachability analysis reduces to finding the maximum output along the normal directions \cite{dutta2018learning}. Some approaches to solve this problem approximate the overall network using polynomials \cite{dutta2019reachability, huang2019reachnn}. Other methods abstract the input-output relationship of neurons using intervals \cite{claviere2021safety}, star sets \cite{tran2019star}, Taylor models \cite{ivanov2021verisig}, polynomial zonotopes \cite{kochdumper2022open}, hybrid zonotopes \cite{zhang2022reachability}, and integration of Taylor models with zonotopes \cite{schilling2022verification}.
%
%The complexity of closed-loop verification is twofold. Imprecisions of the reachability analysis of each step can accumulate, rendering the results too conservative after only a couple of time steps \cite{neumaier1993wrapping}. 
% To mitigate the wrapping effect, recent works utilize techniques such as polynomial zonotopes \cite{kochdumper2022open}, hybrid zonotopes \cite{zhang2022reachability}, and integrating Taylor models with zonotopes \cite{schilling2022verification}.
\cite{sidrane2022overt} explore the extension to nonlinear systems by abstracting nonlinear functions with a set of optimally tight piecewise linear bounds.
\cite{hu2020reach} over-approximates the reachable sets using semidefinite relaxations, which can be used to propagate ellipsoids in addition to polytopes. \cite{everett2021reachability} use a looser convex relaxation \cite{zhang2018efficient} but bridge the gap by partitioning in the state space. \cite{entesari2022reachlipbnb} propose a BnB framework using Lipschitz bounds. \cite{chen2022one} characterize the conditions for the set propagation algorithm under which unrolling the dynamics yields less conservative bounds on the reachable sets. \cite{rober2022backward} propose a framework for closed-loop verification based on backward reachability analysis.   
% \\
% One-shot reachability is an alternative way to acquire reachable sets. \cite{yang2013one} consider their applicability to differential games and \cite{adzkiya2014backward} model them for autonomous max-plus-linear systems. \cite{chen2022one} characterize the conditions under which one-shot reachability is tighter than recursive reachability.

%The choice of the template sets can have a crucial effect on the accuracy of computations. 
%Any potential mismatch between the shape of the template set and the actual reachable set can lead to conservative bounds. 

\paragraph{Set Representation} Axis-aligned hyperrectangles or oriented hyperrectangles, acquired by running principal component analysis (PCA) on samples trajectories \cite{stursberg2003efficient}, can be convenient choices as the number of facets grows only linearly with the ambient dimension. However, they can become too conservative.
%To improve on this, one can rotated hyperrectangles, the orientation of which can be inferred using principal component analysis (PCA) on sampled trajectories \cite{stursberg2003efficient}. But even this choice can still be imprecise, not capturing the shape of the exact output set, resulting in accumulation of errors in recursive reachability. This can also adversely affect the run time of branch-and-bound methods that branch the input space as they will have to work with larger input sets.
On the other hand, there exist more complex methods like \cite{bogomolov2017counterexample} that solve convex problems to derive template polytope directions or \cite{ben2012reachability} that use a first-order Taylor approximation of $F$ (assuming $F$ is differentiable) to propose dynamic directions. 

%{\color{red} Missing References: Girard et al. about polynomial reachability using adaptive polyedra, Paper by Mirco Giacobbe on template polytopes...}

\subsection{Notation}
For a real number $r$, $(r)_+=\max(r,0)$ is the non-negative part of that number. 
For any vector $x \in \mathbb{R}^n$, $\mathrm{ReLU}(x) = [(x_1)_+, \cdots, (x_n)_+]^\top$ where $x_i$ is the $i$th element of $x$.
$I_n$ denotes the $n$ by $n$ identity matrix. $0_{n \times m}$ denotes the $n$ by $m$ matrix of all zeros.
We use $\mathrm{Poly}(A, d) = \{ x \in \mathbb{R}^n \mid Ax \leq d\}$ to denote polyhedrons, in which we denote the $i$-th row of $A$ by $a_i^\top$. This is called the $\mathcal{H}-$representation of polyhedra as it uses the intersection of half-spaces to represent the set. If such a set is bounded, we call it a polytope.
Given matrices $A_i \in \mathbb{R}^{n \times n_i}$ and vectors $a_i \in \mathbb{R}^{n}$, $[A_1, \cdots, A_m, a_1, \cdots, a_k]$ denotes the horizontal concatenation of the elements into a single matrix. For a given function $f:\mathbb{R}^n \rightarrow \mathbb{R}^n$, $f^{(k)}$ denotes the $k$-the composition of $f$ with itself. %i.e., $f^{(k)} = \underbrace{f \circ f \circ \cdots \circ f}_{\text{k times}}$.

%{\color{red} ReLU has not been defined.}

\section{Problem Statement}

%\subsection{}
%
% Consider a discrete-time autonomous system
% %
% \begin{equation}
%     x^{k+1} = F(x^k) \quad x^0 \in \mathcal{X}^0
% \end{equation}
% %
% where $\mathcal{X}^0 \subset \mathbb{R}^n$ is a compact set of initial conditions. 
Consider a discrete-time autonomous system
\begin{equation}\label{eqn:linearSystem}
    x^{k + 1} = F(x^k) = Ax^k + B f(x^k) + e,
\end{equation}
with state $x^k \in \mathbb{R}^n$, control input $u^k \in \mathbb{R}^m$, and system dynamics
%at time $k$
given by $A \in \mathbb{R}^{n \times n}$, $B \in \mathbb{R}^{n \times m}$, and $e \in \mathbb{R}^{n}$ is a known exogenous constant. We assume that the control policy is given by $f(x^k)$, where $f:\mathbb{R}^{n_0} \rightarrow \mathbb{R}^{n_L}$ ($n_0=n, n_L = m$) is a fully-connected ReLU network as follows
\begin{align}\label{eqn:nnDefinition}
    \begin{split}
        x_0 =x, \quad&
        x_{i + 1} = \mathrm{ReLU}(W_i x_i + b_i), \quad i=0, \cdots, L-1,\quad f(x) 
        = W_{L} x_{L} + b_{L},
    \end{split}
\end{align} 
where $x_i \in \mathbb{R}^{n_i}$, $W_i \in \mathbb{R}^{n_{i + 1} \times n_i}$, and $b_i \in \mathbb{R}^{n_{i + 1}}$.
%
% The closed loop can then be written as
%
%$u^k$ is given by an MPC designed for the open-loop linear system. We train a $\mathrm{ReLU}$ neural network $f$ that approximates the MPC. Using the trained network, system \eqref{eqn:linearSystem} can be written as 
% \begin{equation}\label{eqn:linearSystem}
%     x^{k + 1} = F(x^k) := Ax^k + Bf(x^k) + e.
% \end{equation}
%
We assume that the initial state $x^0$ belongs to a bounded box $\mathcal{X}^0$ of initial conditions, $\mathcal{X}^0 = \{x \in \mathbb{R}^n \mid \underline{x} \leq x \leq \bar{x}\}$, where $\underline{x},\bar{x} \in \mathbb{R}^{n}$. Starting from this set, the reachable set at time $k+1$ is defined as in \eqref{eqn:reachabeSetDefinition}.
For piecewise linear dynamics and polyhedral initial sets, these sets are non-convex. 

In this paper, we use bounded polyhedra to parameterize the template sets,  $\bar{\mathcal{X}}^{k}= \mathrm{Poly}(C^{k},d^k)$,
%
% \begin{align}
% \bar{\mathcal{X}}^k = \mathrm{Poly}(C^{k},d^k) = \{x \mid C^k x \leq d^k \}   
% \end{align}
%
where the rows of the template matrix $C^k \in \mathbb{R}^{m_k \times n}$ are the normal directions 
% of the $m_k$ half-spaces
that define the facets of the  polyhedron and $d^k \in \mathbb{R}^{m_k}$ is their offset. We use the superscript $k$ for the template matrices to emphasize that the template may not be the same for all the polyhedra. Starting from $\bar{\mathcal{X}}^0={\mathcal{X}}^0$, our goal is to compute the pair $(C^k,d^k)$ 
% (including the number of facets $m_k$)
such that $\bar{\mathcal{X}}^k$ over-approximates ${\mathcal{X}}^k$ as closely as possible.

%In this paper, our goal is to compute polyhedral over-approximations of the reachable sets using adaptive template polytopes that can capture the shape of the reachable sets by exploiting the structure of the closed-loop map.

%In this paper, our goal is to exploit the structure of the closed-loop system in order to choose template matrices $C^k$ that can capture the shape of the reachable set $\mathcal{X}^k$ as much as possible. Given the template matrix, our second goal is to to compute the offset vectors $d^k$ to ensure that $\bar{\mathcal{X}}^k$ has minimial volume and encloses $\mathcal{X}^{k}$. 
%

\section{Proposed Method} \label{section: adpative template polytope}
In this section, we present the details of the proposed method. We first cast the closed-loop map as an equivalent $\mathrm{ReLU}$ network ($\S$\ref{subsec::Neural Network Representation of the Closed-loop Map}). Then  assuming that the template matrices $C^k$ are all given, we elaborate on computing the offset vectors $d^k$ ($\S$\ref{sec:bnbDeepPoly}). We will finally discuss how to choose the template matrices $C^k$, the choice of which will not depend on $d^k$ ($\S$\ref{subsec::Adaptive Template Polytopes}).

\subsection{Neural Network Representation of the Closed-loop Map} \label{subsec::Neural Network Representation of the Closed-loop Map}
% Start with unrolling? $\mathcal{X}^{k + 1} = \{ F^{(k + 1)}(x) | x\in \mathcal{X}^0$, where $F^{(k + 1)}$ is the $(k + 1)$th iterate of F. This is a residual network. In order to make use of state of the art verification methods, we propose a method a to convert it to a sequential model without a residual connection.

% Does this help? $F(x) = A(x - l) + Bf(x) + l + Al$.

% Remove "what this neural network is doing". Essentially we are creating an alternative path for the input to pass through the neural network unchanged.
% \begin{figure}
%     \centering
%     \includegraphics[width=.5\textwidth]{images/skipConnection.pdf}
%     \caption{closed-loop system \eqref{eqn:linearSystemSystem} as a neural network with a skip connection. \red{replace f. combine images}.}
%     \label{fig:skipConnection}
% \end{figure}

% Consider the linear feedback system
% \begin{equation}
%     x^{k + 1} = Ax^k + Bu^k + e,
% \end{equation}
% with state $x^k \in \mathbb{R}^n$, control input $u^k \in \mathbb{R}^m$, and system dynamics given by $A \in \mathbb{R}^{n \times n}$, $B \in \mathbb{R}^{n \times m}$, and $e \in \mathbb{R}^{n}$. We assume that the feedback control policy is given by a $\mathrm{ReLU}$ neural network $f$:
% \begin{equation}\label{eqn:linearSystemSystem}
%     x^{k + 1} = f_{cl}(x) = Ax^k + Bf(x^k) + e,
% \end{equation}
% %
% Given a polyhedron $\mathcal{X}$, our goal is to compute a polyhedral over-approximation of $f_{cl}$. In principle, we can first compute the polyhedral over-approximation of the maps $Ax+e$ and $Bf(x)$ separately. 
%
We start by observing that the closed-loop map is essentially a neural network with a skip connection as shown in Fig. \ref{fig:nns}-(a). State-of-the-art bound propagation methods for neural network verification operate on sequential neural networks, i.e., neural networks without any skip connections. To take advantage of existing bound propagation methods with minimal intervention, we convert the closed-loop map into an equivalent sequential network without skip connections.
%As a result, they are not adept to handle skip connections. This raises a concern as the closed-loop map $F$ of the system, as portrayed in figure \ref{fig:nns}.a, is a neural network with a skip connection. Thus, it is imperative that to take advantage of the readily available state-of-the-art methods, we have to reformulate \eqref{eqn:linearSystem} as a sequential network.
% DeepPoly, as presented in \cite{singh2019abstract}, operates on sequential neural networks, i.e., neural networks without any skip connections.
% However, the closed-loop map $F$, as portrayed in figure  in \ref{fig:nns}.a, is a neural network with a skip connection. %
% As a result, to take advantage of the DeepPoly framework, it is imperative to reformulate \eqref{eqn:linearSystem} as a sequential neural network.
%
%Using a sequential network implies that there is a single flow route in the network, and thus, the activation function applied element-wise to the variables at a layer $j$, is the same for all such variables at that layer. This causes a dilemma as we need to pass the input $x$ through the network untouched, until $f(x)$ has been calculated, as the $\mathrm{ReLU}$ activation functions would cut-off parts of the input set $\mathcal{X}^k$ that do not lie in the non-negative quadrant. Our proposed modified neural network overcomes this by shifting the entire input set to the nonnegative quadrant. 
To be precise, note that we can write $x^k = \mathrm{ReLU}(x^k) - \mathrm{ReLU}(-x^k)$. This implies that $x^k = (\mathrm{ReLU} \circ \cdots \circ \mathrm{ReLU})(x^k) - (\mathrm{ReLU} \circ \cdots \circ \mathrm{ReLU})(-x^k)$, i.e., $x^k$ can  be routed through a sequential $\mathrm{ReLU}$ network. By concatenating this path on top of  $f$, we obtain the following equivalent representation of $F$,
%for an input $x^k \in \mathcal{\bar{X}}^k \subset \mathbb{R}^{n_0}$, where the set $\mathcal{\bar{X}}^k$ is bounded below by $l^k$, we have: 
% The closed-loop map is essentially a neural network with a skip connection, as figure \ref{fig:skipConnection}. To use DeepPoly to find an overapproximation of the reachable set of $\mathcal{X}^{k + 1}$ we need to reformulate equation \eqref{eqn:linearSystemSystem} as a sequential neural network, i.e., a neural network with no skip or residual connections. However, if we were to simply pass the input $x^k$ through the network by simply stacking a second flow route of linear layers, the $\mathrm{ReLU}$ nonlinearities, that have to be the same for both flow routes (as we are limited to sequential models), would cut-off parts of the input set $\mathcal{X}^k$ that don't lie in the non-negative real quadrant. To overcome this, we propose a modified neural network. For an input $x^k \in \mathbb{R}^{n_0}$ bounded below by $l^k \leq x^k$ as follows:
\begin{align}\label{eqn:modifiedNn}
    \begin{split}
        y_0 &= \begin{bmatrix}
        I_{n_0}\\ -I_{n_0} \\ I_{n_0}
        \end{bmatrix}x^k + \begin{bmatrix}
        0_{{n_0} \times 1}\\ 0_{{n_0} \times 1} \\0_{{n_0} \times 1}
        \end{bmatrix}, \quad y_i = \mathrm{ReLU}(\hat{y}_i), \quad i = 1, \cdots, L,\\
        \hat{y}_{i + 1} &= \begin{bmatrix}
        I_{n_0} & 0_{{n_0} \times n_0} & 0_{{n_0} \times n_i}\\
        0_{{n_0} \times n_0} & I_{n_0} & 0_{{n_0} \times n_i}\\
        0_{n_{i + 1} \times n_0} & 0_{n_{i + 1} \times n_0} & W_i
        \end{bmatrix}y_i + 
        \begin{bmatrix}
        0_{n_0 \times 1}\\ 0_{n_0 \times 1} \\ b_i
        \end{bmatrix}, \quad i=0, \cdots, L - 1,\\ 
        y_{L + 1} &= \begin{bmatrix}
        I_{n_0} & 0_{{n_0} \times n_0} & 0_{{n_0} \times n_L}\\
        0_{{n_0} \times n_0} & I_{n_0} & 0_{{n_0} \times n_L}\\
        0_{n_{L + 1} \times n_0} & 0_{n_{L + 1} \times n_0} & W_L
        \end{bmatrix}y_L + 
        \begin{bmatrix}
        0_{n_0 \times 1}\\ 0_{n_0 \times 1} \\ b_L
        \end{bmatrix}, x^{k + 1} = \begin{bmatrix}
        A & -A & B
        \end{bmatrix} y_{L + 1} + e.
    \end{split}
\end{align}
\begin{figure}[t]
    \centering
    \includegraphics[width=0.9\textwidth]{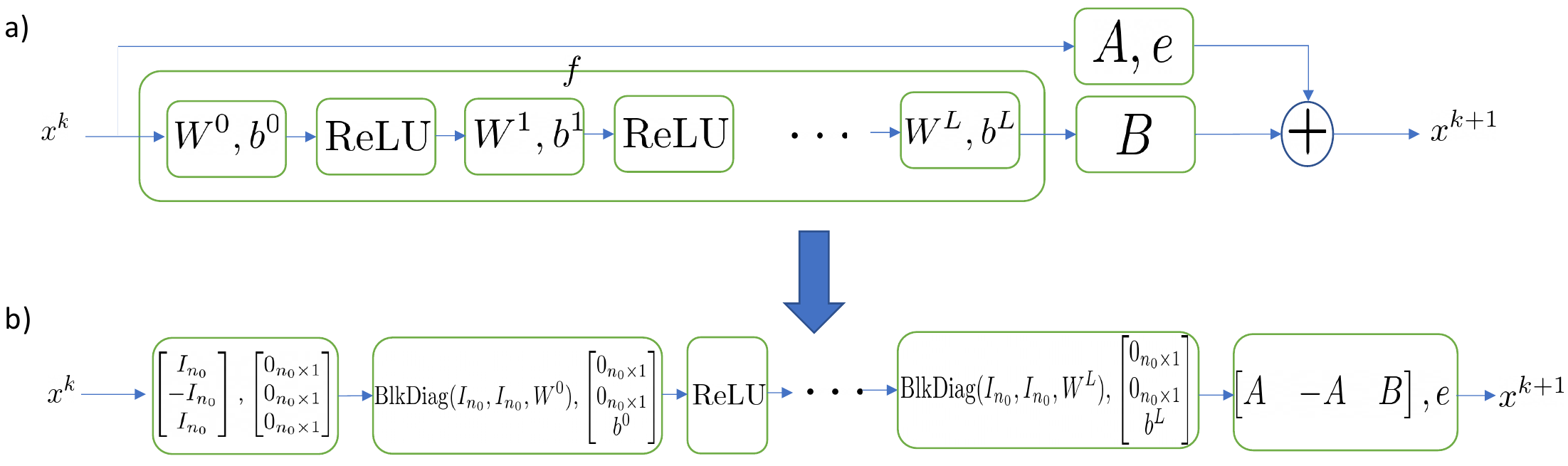}
    \caption{Closed-loop system block diagrams. The first matrix in each block represents the weights matrix of that linear layer, while the second matrix represents the bias of that layer. a) The original system as proposed by \eqref{eqn:linearSystem}. b) Conversion of the system dynamics into an equivalent $\mathrm{ReLU}$ network, as explained by \eqref{eqn:modifiedNn}.}
    \label{fig:nns}
\end{figure}
% By writing the closed-loop map in the form \eqref{eqn:modifiedNn}, we have essentially created an alternative path for the state $x^k$ to pass through the network untouched. 
%This works by first shifting the input set to the non-negative quadrant and then cancelling this shift when the original variable is needed. 
We use $F_{eq}$ to denote the neural network in \eqref{eqn:modifiedNn} (see Fig. \ref{fig:nns}-(b)). 
%The superscript $k$ is to emphasize the dependence of $F_{eq}^k$ on $k$, as it is parameterized by the lower bound vector $\ell^k$. %It is worth noting that the lower bound  $l^k$ can be any valid lower bound on the reachable set $\mathcal{X}^k$ and does not need to be a tight approximation of the actual lower bound. 
In summary, we can rewrite the closed-loop system \eqref{eqn:linearSystem} equivalently as
\begin{align}\label{eqn:equivalentFormulation}
    x^{k + 1} = Ax^k + Bf(x^k) + e = F_{eq}(x^k). %\quad %\forall x^k \in \bar{\mathcal{X}}^k.
\end{align}
Formulating the closed-loop system in the form of \eqref{eqn:modifiedNn} is also useful in that it eliminates the need to compute the  Minkowski sum of the two polytopes enclosing $A\bar{\mathcal{X}}^k + e$ and $Bf(\bar{\mathcal{X}}^k)$. 
%That is, given equation \eqref{eqn:linearSystem} and the two bounding polytopes for the terms $Ax^k + e$ and $Bf(x^k)$, to find a polytope bounding $x^{k + 1}$ we would still need to calculate the minkowski sum of the two polytopes. 
Not only has the process of finding the Minkowski sum of two polytopes proven to be challenging \cite{althoff2021set}, this would amount to the unnecessary enlargement of the approximated reachable set $\mathcal{\bar{X}}^{k + 1}$, resulting in less accurate overapproximations of the actual reachable sets. We will subsequently work with the equivalent representation \eqref{eqn:equivalentFormulation}.

\subsection{Solving the Optimization Problem Using Branch and Bound}\label{sec:bnbDeepPoly}

Suppose that the polytope $\bar{\mathcal{X}}^k$ has been computed and that the template matrix $C^{k+1}$ is given. To over-approximate $F(\bar{\mathcal{X}}^k)$ by $\bar{\mathcal{X}}^{k+1} = \mathrm{Poly}(C^{k+1},d^{k+1})$, the offset vector $d^{k+1}$ of the minimal volume $\bar{\mathcal{X}}^{k+1}$ enclosing $F(\bar {\mathcal{X}}^k)$ must satisfy
\begin{align} \label{eq: computation of effset vector}
    % \bar{\mathcal{X}}^k 
     d_i^{k+1} = \sup \{ {c_i^{k+1}}^\top F(x) \mid x \in \bar{\mathcal{X}}^  {k}\}
     =\sup \{ {c_i^{k+1}}^\top F_{eq}(x) \mid x \in \bar{\mathcal{X}}^{k}\} = 
     \quad i=1,\cdots,m_{k+1}.
\end{align}
where the second equality follows from the equivalence in \eqref{eqn:equivalentFormulation}, i.e., $F(x)=F_{eq}^k(x)$ for all $x \in \mathbb
R^{n_x}$.
Computation of the maximal values in \eqref{eq: computation of effset vector} involves solving an MILP with the binary variables describing the $\mathrm{ReLU}$ neurons \cite{tjeng2017evaluating}. 
Generic MILP solvers may not efficiently utilize the underlying problem structure. Therefore, custom BnB methods aim to exploit highly specialized bound propagation methods in the bounding stage. However, these bound propagation techniques could prove inefficient for general convex polyhedra input sets, such as the polytope $\bar{\mathcal{X}}^k$, as they require solving linear programs (LPs) to compute the bounds. Furthermore, general convex polytopes do not lend themselves to efficient partitioning. 

To overcome these obstacles, we propose an end-to-end reachability analysis on the unrolled dynamics, following a similar approach to \cite{chen2022one}. Observing that the reachable set at time $k+1$ can be written as $\mathcal{X}^{k+1} = \{F^{(k+1)}(x) \mid x \in \mathcal{X}^0\}$, we  can directly over-approximate $\mathcal{X}^{k+1}$ by $\bar{\mathcal{X}}^{k+1} = \mathrm{Poly}(C^{k+1},d^{k+1})$, where $d^{k+1}$ is now obtained by solving
\begin{align} \label{eq: computation of effset vector unrolled dyn}
    d_i^{k+1} \!=\!  \sup \{ {c_i^{k+1}}^\top F^{(k+1)}(x) \!\mid \!  x \in {\mathcal{X}}^{0}\} \!=\! \sup \{ {c_i^{k+1}}^\top F_{eq}^{(k + 1)}(x) \!\mid \! x \in {\mathcal{X}}^{0}\} \ i=1,\cdots,m_{k+1}.
\end{align}
%
%To obtain the second identity,  we have defined $F_{eq}^{(k+1)} = F_{eq}^k \circ \cdots \circ F_{eq}^0$, where each $F_{eq}^i$ is an equivalent representation of $F$ over $\mathcal{X}^i$, i.e., $F(x)=F_{eq}^k(x)$ for all $x \in \bar{\mathcal{X}}^k$. 
Thus we have converted the unrolled dynamics $F^{(k+1)}$, which is essentially a neural network with skip connections, to an equivalent neural network $F_{eq}^{(k+1)}$ without any skip connections.

Based on the assumption that $\mathcal{X}^0$ is a hyperrectangle, we propose to solve each of the optimization problems in \eqref{eq: computation of effset vector unrolled dyn} using a BnB method based on recursive partitioning of ${\mathcal{X}}^0$, and bounding the optimal value $d_i^{k+1}$ efficiently over the partition. To be precise, given any sub-rectangle $\mathcal{X} \subset \mathcal{X}^0$, the bounding subroutine produces lower $\underline{d}_i({\mathcal{X}})$ and upper bounds $\overline{d}_i({\mathcal{X}})$ such that $\underline{d}_i({\mathcal{X}}) \leq \sup \{ {c_i^{k+1}}^\top F_{eq}^{(k + 1)}(x) \mid x \in {\mathcal{X}}\} \leq \overline{d}_i({\mathcal{X}}).$
%
% \[
% \underline{d}_i({\mathcal{X}}) \leq \sup \{ {c_i^{k+1}}^\top F_{eq}^{(k + 1)}(x) \mid x \in {\mathcal{X}}\} \leq \overline{d}_i({\mathcal{X}}).
% \]
%
Now, given a rectangular partition of the input set $\mathcal{X}^0$ and the corresponding bounds at each iteration of the BnB algorithm, we can update the lower (respectively, upper) bound on $d_i^{k+1}$ by taking the minimum of the lower (respectively, upper) bounds over the partition of $\mathcal{X}^0$. For the next iteration, the partition is refined non-uniformly based on a criterion, and those sub-rectangles that cannot contain the global solution are pruned. Overall, the algorithm produces a sequence of non-decreasing lower and non-increasing upper bounds on the objective value, and under mild conditions, it terminates in a finite number of iterations with a certificate of $\epsilon$-suboptimality (see \cite{boyd2007branch} for an overview of BnB methods).
%
%By recursively partitioning the set $\mathcal{X}^0$ and discarding those subrectangles that cannot contain the global solution, we can bound the optimal value of \eqref{eq: computation of effset vector unrolled dyn 1} within an arbitrary accuracy. 
%
%

The above framework in the context of the problem presented in this paper has been developed recently in \cite{entesari2022reachlipbnb}, in which the bounds over the partition are computed by utilizing upper bounds on the Lipschitz constant of the objective function (here ${c_i^{k+1}}^\top F_{eq}^{(k+1)}$). 
Although any bound propagation method can be incorporated into our framework (on account of the conversion in \eqref{eq: computation of effset vector unrolled dyn}), in this paper, we use the fast and scalable method of DeepPoly \cite{singh2019abstract} within the branching strategy of \cite{entesari2022reachlipbnb}. DeepPoly can handle larger networks and consequently, longer time horizons in our case. For details of the branching strategy, see \cite{entesari2022reachlipbnb}.

\begin{remark}[Zonotope Initial Set] \normalfont
    Throughout the paper, we focus on hyper-rectangle input sets $\mathcal{X}^0$. However, it is fairly simple to adapt our method to zonotope initial sets $\mathcal{X}^0 = \lbrace x \mid x = Gz, \|z\|_{\infty}\leq 1\}$. To do so,
    we substitute $x=Gz$ in \eqref{eqn:linearSystem} and absorb $G$ into the first linear layer of $f$, i.e., use $W^0G$ in lieu of $W^0$ in \eqref{eqn:nnDefinition}. We then apply the BnB method on this modified network and the input set $\mathcal{Z} = \lbrace z \mid \|z\|_{\infty}\leq 1\rbrace$.
\end{remark}

\subsection{Adaptive Template Polytopes}
\label{subsec::Adaptive Template Polytopes}

As discussed previously, the choice of the template sets can have a crucial effect on the accuracy of computations. The advantage of using fixed template polyhedra is that geometric operations such as union or intersection can be performed more efficiently \cite{althoff2021set}. However, choosing a flexible template a priori that can capture the shape of the reachable sets is challenging. 
%
%
%how can we capture the shape of the reachable set $\mathcal{X}^{k + 1} = F(\bar{\mathcal{X}}^k)$ as accurately as possible without using too many facets?
In this subsection, we propose a method  that can adapt the geometry of the template polyhedra to the shape of the reachable sets. We outline the building blocks of the method for a single $\mathrm{ReLU}$ plus affine layer. Extension to a full neural network is straightforward and is provided in Algorithm \ref{alg:adaptivePolytopeDirections}.

%obtain suitable directions for a polytope by which we want to localize the reachable set of $x^{t + 1}$. We first present our method for a fully-connected network $f$ with $\mathrm{ReLU}$ nonlinearities. In Section \ref{}, we will extend this approach to the closed-loop map $F(x) = A x + Bf(x)$, which is essentially a feed-forward network with skip connection. 

\paragraph{Affine Layers} Consider an affine layer $ y = W x + b$, where $W \in \mathbb{R}^{n_1 \times n_0},b \in \mathbb{R}^{n_1}$, and suppose $x \in \mathcal{X} = \mathrm{Poly}(C,d)$,   $C \in \mathbb{R}^{m \times n_0}$.
The image of $\mathcal{X}$ under the affine layer is given by 
$$\mathcal{Y} = W\mathcal{X} + b = \{ y \in \mathbb{R}^{n_1} \mid y = W x +b, Cx \leq d\}.$$
In principle, we can find an exact $\mathcal{H}$-representation of $\mathcal{Y}$ using quantifier elimination methods such as Fourier-Motzkin elimination or Ferrante and Rackoff's method \cite{bradley2007calculus}. However, these methods typically have doubly exponential complexity. For example, using Ferrante and Rackoff's method, the complexity is $O(2^{2^{pn_0}})$ for some fixed constant $p$.
Here we propose a heuristic based on singular value decomposition of $W$ to compute an efficient while good approximation of the shape of the output polytope $\mathcal{Y}$. 
%   '
%
%We want to find a matrix $\hat{A}$ that would have the suitable directions mimicking that of the actual polytope that constrains $y$. 

Assuming that $W \in \mathbb{R}^{n_1 \times n_0}$ has rank $r \leq \min(n_1,n_0)$, its singular value decomposition can be written as $W = U\Sigma V^\top =  \sum_{i = 1}^{r} \sigma_i u_i v_i^\top$, where $\sigma_i$, $u_i$ and $v_i$ are the $i^{th}$ largest singular value, and the corresponding left and right singular vectors, respectively. 
We also note that the range space of $W$ is given by $\mathrm{Span}\lbrace u_1, \cdots, u_r \rbrace$ and the null space is given by $\mathrm{Span}\lbrace v_{r + 1}, \cdots, v_{n_0} \rbrace$:
\[
Wx = \sum_{i = 1}^{r} (\sigma_i  v_i^\top x) u_i, \quad \text{and} \quad  W v_j = \sum_{i = 1}^{r} \sigma_i u_i (v_i^\top v_j) = 0 \quad j=r+1,\cdots,n_0. %= \sum_{i = 1}^{r} a_i u_i, \quad a_i = \sigma_i v_i^\top x.
\]
We distinguish between two cases: 
%\paragraph{Fat weight matrix ($n_1 \leq n_0$):}
%\begin{itemize}[leftmargin=*]
\paragraph{Fat Weight Matrix ($n_1 \leq n_0$):} 
Consider the system of equations $Wx + b = y$. Solving for $x$, with the knowledge that it is solvable, yields $x = W^\dagger (y - b) + F\zeta$, where the columns of $F \in \mathbb{R}^{n_0\times (n_0 - r)}$ are the vectors $v_{r + 1}, \cdots, v_{n_0}$, i.e., they span the null space of $W$, and $\zeta \in \mathbb{R}^{(n_0 - r)}$ is a free variable. This can be viewed as $x = \begin{bmatrix} 
W^\dagger& F
\end{bmatrix} 
\begin{bmatrix}
y\\
\zeta
\end{bmatrix} - W^\dagger b$.
Given $C\begin{bmatrix}
W^\dagger& F
\end{bmatrix} 
\begin{bmatrix}
y\\
\zeta
\end{bmatrix} - CW^\dagger b \leq d$,
% \begin{equation}\label{eqn:boundedPreProjection}
%     C\begin{bmatrix}
% W^\dagger& F
% \end{bmatrix} 
% \begin{bmatrix}
% y\\
% \zeta
% \end{bmatrix} - CW^\dagger b \leq d
% \end{equation}
%
we are interested in bounding $y$. This is equivalent to projecting the polytope defined in the $(y, \zeta)$ space to the $y$ space and bounding it.
% , what we want is the projection of the bounded body defined by \eqref{eqn:boundedPreProjection} in the $(y, \zeta)$ space, to the $y$ space. 
%
To do so, since $\zeta$ is a free variable, we absorb it into the constant term $CW^\dagger y \leq d + CW^\dagger b - CF\zeta = d^\prime$.
As a result, we propose to use the template matrix $\hat{C} = CW^\dagger$ and leave $d'$ unspecified. %solve our desired optimization problems to find the corresponding values of $d^\prime$.
    
% \begin{figure}%
%     \centering
%     \subfloat{{\includegraphics[width=5cm]{images/poly3.eps} }}%
%     \qquad
%     \subfloat{{\includegraphics[width=5cm]{images/poly4-2-1.eps} }}%
%     \caption{2 Figures side by side}%
%     \label{fig:example}%
% \end{figure}

\begin{figure}[t!]
    \centering
    \includegraphics[width=.3\textwidth]{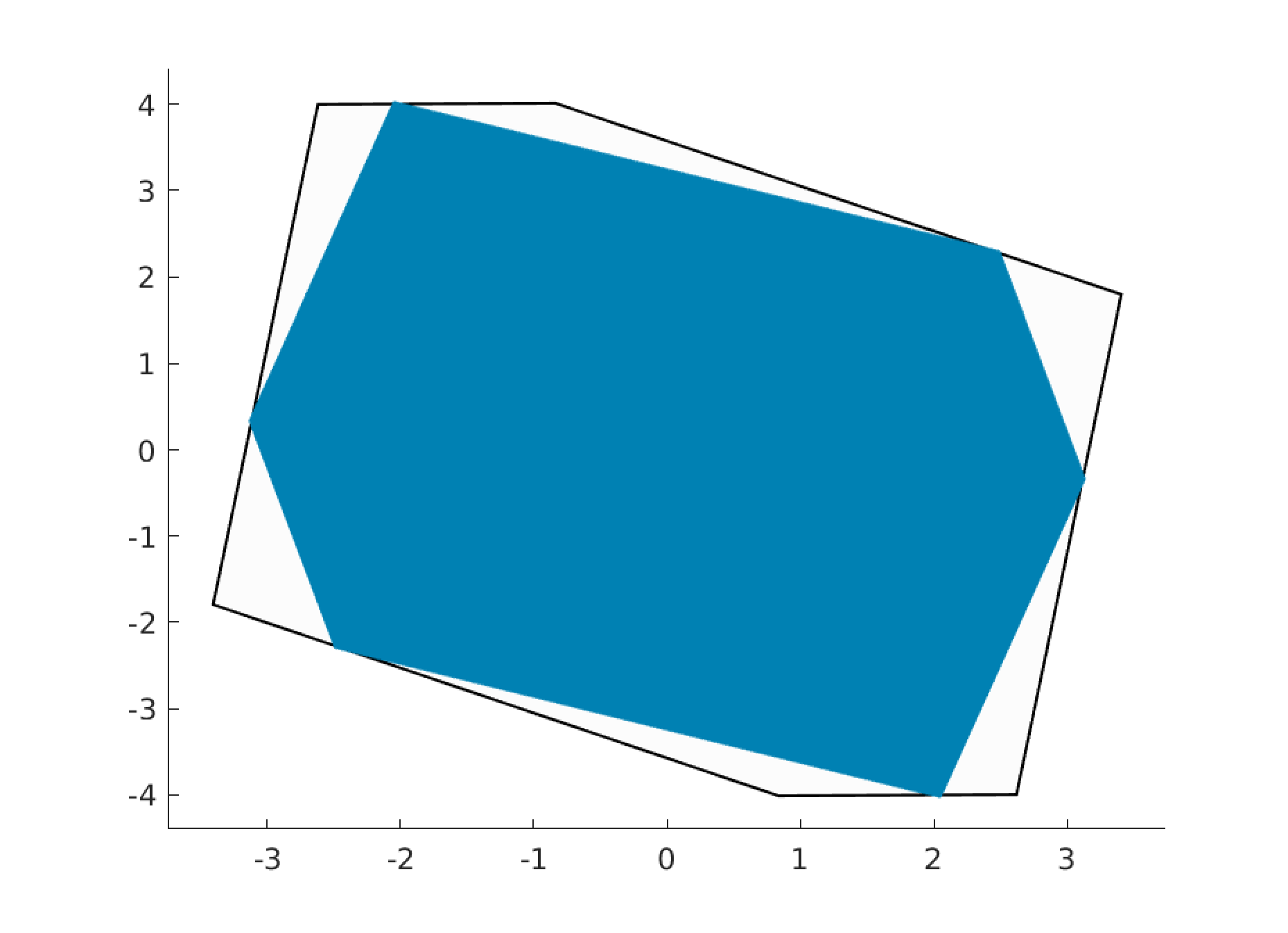}
    \includegraphics[width=.3\textwidth]{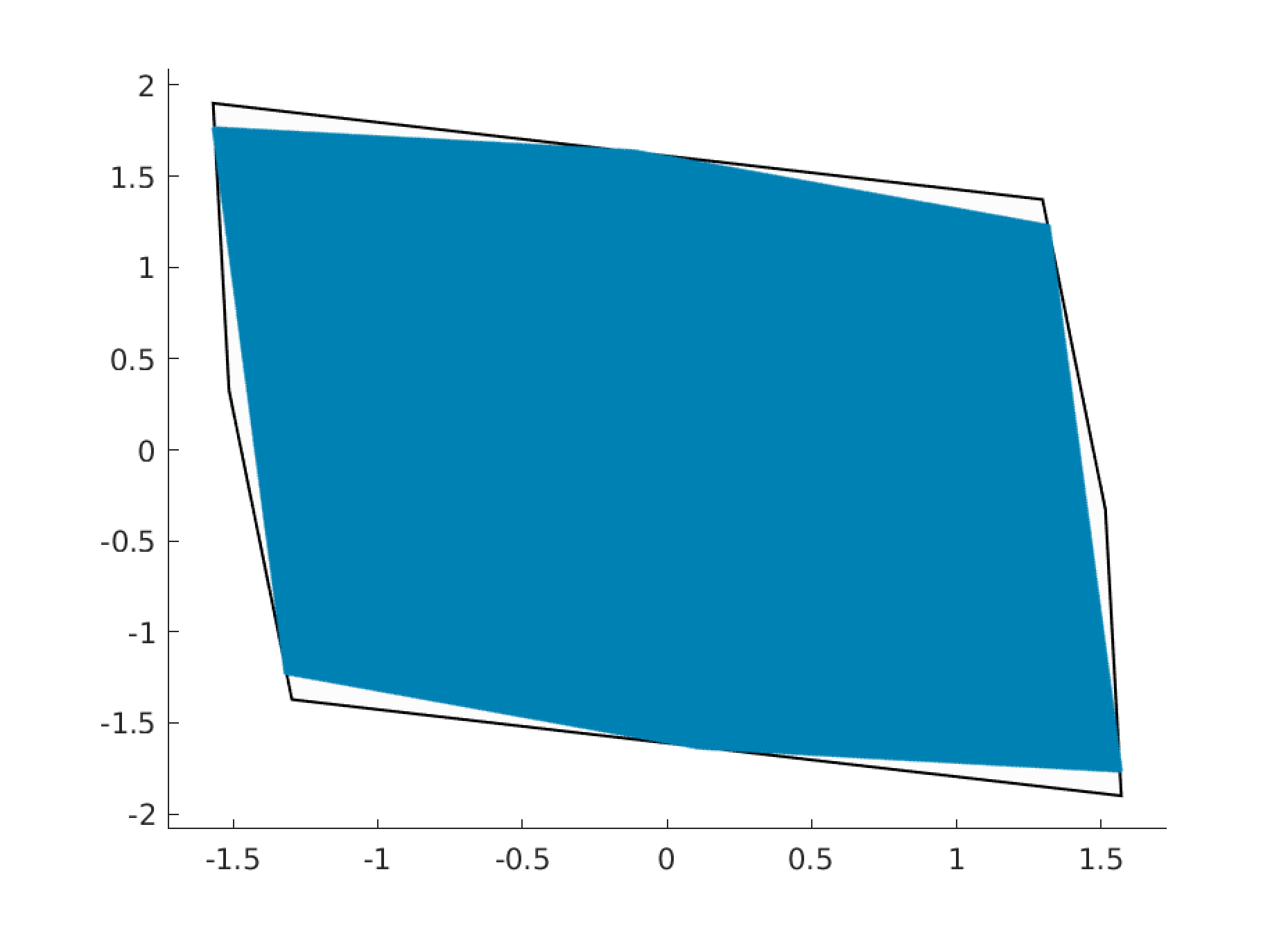}
    \includegraphics[width=.3\textwidth]{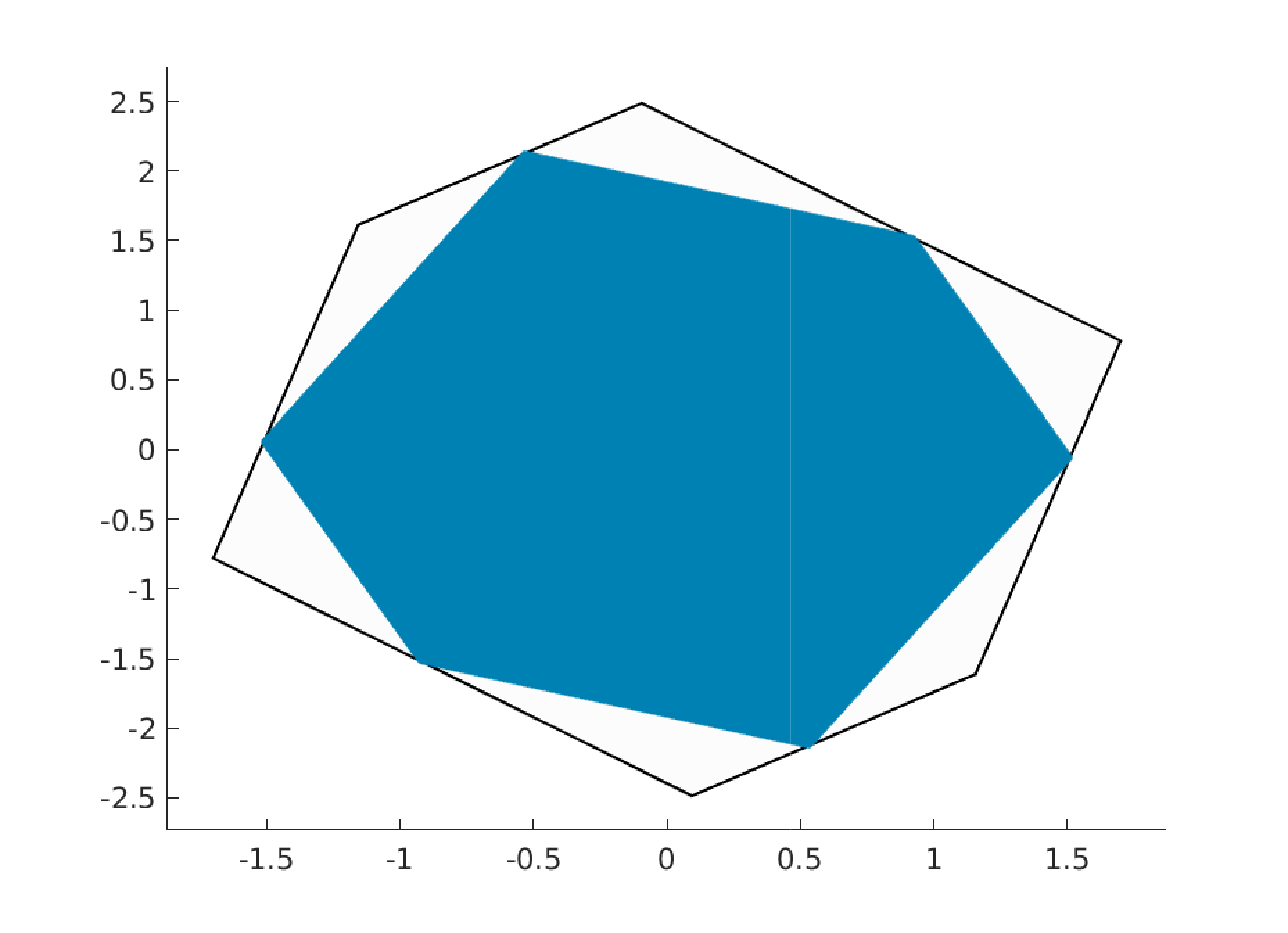}
    \captionof{figure}{Examples of the polytopes acquired by solving the directions provided by the adaptive polytope method for affine layers with $n_1 = 2, n_0 = 3$. The blue set is the exact image of the transformation whilst the bounding polytope is acquired by solving linear programs $\min_{Cx \leq d} c^\top (Wx + b)$ for $c$s taken as the rows of $CW^\dagger$.}
    \label{fig:sampleAdaptivePolytopes}
\end{figure}

\paragraph{Tall Weight Matrix ($n_1 > n_0$):} In this scenario, assuming $W$ has full rank, the system of equations $Wx + b = y$ has a unique solution given by $x = W^\dagger (y - b)$. 
We now have
\[
CW^\dagger(y - b) = Cx \leq d \Rightarrow CW^\dagger y \leq d + CW^\dagger b = d^\prime.
\]
As a result, we propose the same directions $CW^\dagger$ as before. Since $u_i^\top W =0$, $i=n_0+1,\cdots,n_1$, we have $u_i^\top y =u_i^\top b$. This shows that $y$ has zero variance in the direction of $u_i, i=n_0+1,\cdots,n_1$. At the same time, $CW^\dagger u_i = 0$. Thus, the directions in $CW^\dagger$ have no components in the space of $\mathrm{Span}\{ u_{n_0 + 1}, \cdots, u_{n_1}\}$. Therefore, we simply add $\pm u_i$ for $i = n_0+1, \cdots, n_1$ to the list of directions, i.e.,
$
\hat{C} = \begin{bmatrix}
(CW^\dagger)^\top,
u_{n_0 + 1},
-u_{n_0 + 1},
\hdots,
u_{n_1},
-u_{n_1}
\end{bmatrix}^\top
$.\\
An interesting feature in this scenario is that the polytope defined by $\hat{C} y \leq d^\prime$ is tight, i.e., all points $y$ that reside in this polytope have a corresponding $x = W^\dagger(y - b)$ that satisfies $Cx \leq d$. To see this, it suffices to realize that for all $y$ that $\hat{C}y \leq d^\prime$, by the construction of $\hat{C}$, we know that $y$ is in the range space of $W$ shifted by $b$. Thus there exists a corresponding $x = W^\dagger(y - b)$. Now, $Cx = CW^\dagger(y - b) \leq d^\prime - CW^\dagger b = d$.
Thus, the polytope is tight and the directions of $\hat{C}$ exactly correspond to the directions of the polytope.
%\end{itemize}

\paragraph{$\mathrm{\mathbf{ReLU}}$ Layers} Consider a $\mathrm{ReLU}$ layer $y = \mathrm{ReLU}(x)$, $x\in \mathbb{R}^{n}$, and suppose $x \in \mathcal{X} = \mathrm{Poly}(C,d)$,   $C \in \mathbb{R}^{m \times n}$. Since $y \geq 0$, the $\mathrm{ReLU}$ activations might cut off the polytope. Thus we can improve the set of facets by adding the corresponding directions
$
    \tilde{C} = \begin{bmatrix}
    C^\top, 
    -e_1,
    \hdots,
    -e_n
    \end{bmatrix}^\top
$, where $e_i$ is the ith standard basis vector. 

Next, we will prove that given a polytope, this proposed method will return a polytope.
\begin{proposition}
Given a full-rank matrix $W$, and a polytope with facets given by $C$, the adaptive polytope method outlined in Algorithm  \ref{alg:adaptivePolytopeDirections} will yield a polytope, i.e., a bounded polyhedron.
\end{proposition}

\begin{proof}
As the $\mathrm{ReLU}$ layers only add a set of facets, they do not affect the bounded property of the polytope as it is already bounded. Thus, we will show that the affine layers preserve the bounded property. Suppose we have $y = Wx + b$ and $x \in \mathrm{Poly}(C,d)$ is bounded. We consider two cases:

% First, note that since $Ax \leq d$ is a polytope, $A$ has full column rank. If not, there exists a direction vector $u$ with $Au = 0$. Take any point $x$ from the polytope and consider $x(\alpha) = x + \alpha u$. As $\alpha$ grows larger, this point will shoot off to infinity, but $Ax(\alpha) \leq d$ still holds. So, $A$ has full column rank. We now consider the two cases:
\begin{enumerate}[leftmargin=*]
    \item $n_1 \leq n_0$: If $CW^\dagger y \leq d^\prime$ is unbounded, there exists a point $y$ and direction $v$ such that $y(\alpha) = y + \alpha v$ tends to infinity as $\alpha$ grows larger. The corresponding point in the $x$ space is given by $x(\alpha) = W^\dagger (y + \alpha v - b) + F\zeta = \alpha W^\dagger v + W^\dagger (y - b) + F\zeta$. Since the input is bounded in a polytope, we must have $W^\dagger v = 0$. The null space of $W^\dagger$ is given by $\mathrm{Span}\{u_{r + 1}, \cdots, u_{n_1}\}$ where $r$ is the rank of $W$. But this can not happen since $W$ has full rank $r = n_1$.
    \item $n_1 > n_0$: Similarly, since we add the directions not spanned by $CW^\dagger$ to the set of facets, the same argument will apply and the suggested directions define a bounded polytope.
\end{enumerate}
\vspace{-8.35mm}
\end{proof}

% \begin{algorithm}
% \caption{Adaptive Polytope Directions}\label{alg:adaptivePolytopeDirections}
% \textbf{Input:}  \\
%  \textbf{Output:} \\
%  \textbf{Initialize:} 
% \begin{algorithmic}
% \FORALL{$i = 0, \cdots, L$}
%  \State$n_1, n_0$ = shape($W^i$)
% \State$U, \Sigma, V = \text{SVD}(W^i)$, \hfill \Comment{$U = \lbrack u_1, \cdots, u_{n_1}\rbrack$}
%  \State$D = D (W^i)^\dagger$
%  \IF {$n_1 > n_0$}
% \State $ D = \lbrack D^\top, u_{n_0 + 1}, -u_{n_0 + 1}, \cdots, u_{n_1}, -u_{n_1}\rbrack^\top $ 
% \ENDIF
% \State $D = \lbrack D^\top, -e_1, \cdots, -e_{n_1}\rbrack^\top$, \hfill \COMMENT{$e_i \in \mathbb{R}^{n_1}$. After each layer, there is a $\mathrm{ReLU}$ nonlinearity.}
% \ENDFOR
% \State Remove similar directions from $D$  according to \eqref{eqn:discardSimilar}
% \RETURN D
% \end{algorithmic}

% \end{algorithm}

% \RestyleAlgo{ruled}

%% This is needed if you want to add comments in
%% your algorithm with \Comment

\SetKwComment{Comment}{/* }{ */}

\begin{algorithm}[b!]

\caption{Adaptive Polytope Directions}\label{alg:adaptivePolytopeDirections}
\KwData{Input polytope template matrix $C$, piecewise linear sequential neural network weights $( W^0, \cdots, W^L)$, similarity tolerance level $\lambda$.}
\KwResult{Output polytope template matrix $D$.}
$D \gets C$\;

\For{$i=0, \cdots, L$}{
$n_1, n_0$ = shape($W^i$) \;

$U, \Sigma, V = \text{SVD}(W^i)$, \hfill  \Comment{$U = \lbrack u_1, \cdots, u_{n_1}\rbrack$}

$D = D (W^i)^\dagger$\;

\If{$n_1 > n_0$}{
$ D = \lbrack D^\top, u_{n_0 + 1}, -u_{n_0 + 1}, \cdots, u_{n_1}, -u_{n_1}\rbrack^\top $ 
}

$D = \lbrack D^\top, -e_1, \cdots, -e_{n_1}\rbrack^\top$, \hfill \Comment{\tiny$e_i \in \mathbb{R}^{n_1}$. After each layer, there is a $\mathrm{ReLU}$ nonlinearity.}
}
Remove similar directions from $D$  according to cosine similarity at tolerance level $\lambda$.

\textbf{Return} $D$
\end{algorithm}

In summary, the method proposed above allows us to approximate the directions of the actual output set. Figure \ref{fig:sampleAdaptivePolytopes} portrays sample results of this algorithm on random single linear layers.\\
The output directions $\hat{C} $ of adaptive polytopes on affine layers will have dimensions $(m + (n_1 - n_0)_+) \times n_1$. Having a consequent $\mathrm{ReLU}$ layer will add $n_1$ rows to $\hat{C}$. For an $L$ layer neural network this would yield $m + \sum_{i = 1}^{L - 1} n_i + \sum_{i = 1}^L (n_i - n_{i - 1})_+$  facets.
However, it is very likely that linear layers with $n_1 \leq n_0$ may produce directions that are similar.  To avoid an explosion of the number of facets, we simply remove similar directions. That is, after the algorithm is run for the full neural network, if for two directions $c_1$ and $c_2$ the cosine similarity is larger  than a certain threshold, i.e., $cos(\theta) = \frac{c_1^\top c_2}{||c_1||_2 ||c_2||_2} > \lambda$, we discard one of the directions.
\begin{remark}
\normalfont
Although it is possible to employ Algorithm \ref{alg:adaptivePolytopeDirections} on the equivalent neural network defined by \eqref{eqn:modifiedNn}, we find, empirically, that it is better to use the algorithm to discern the set of directions $\hat C_1$ and $\hat C_2$ for $A\bar{\mathcal{X}}^k$ and $Bf(\bar{\mathcal{X}}^k)$, respectively, and then use $\hat C^\top = \begin{bmatrix}
\hat C_1^\top & 
\hat C_2^\top
\end{bmatrix}$ as the final set of directions for $x^{k + 1} = Ax^k + Bf(x^k) + e$. To ensure the boundedness of the polytope, we require that $A$ or $B$ be full rank. If not, one can simply add the directions corresponding to axis-aligned hyperrectangles.
\end{remark}

% \begin{figure}[t!]
%     \centering
%     % \includegraphics[width=.32\textwidth]{images/doubleIntegratorDeepPoly.jsonIteration1.png}
%     % \includegraphics[width=.32\textwidth]{images/doubleIntegratorDeepPoly.jsonIteration3.png}
%     % \includegraphics[width=.32\textwidth]{images/doubleIntegratorDeepPoly.jsonIteration4.png}
%     \includegraphics[width=\textwidth,  left]{images/doubleIntegrator}
%     \caption{Computation of reachable sets for the double integrator system.}
%     \label{fig:doubleIntegratorReach}
% \end{figure}

\begin{figure}[t!]
    \centering
    \includegraphics[width=0.8\textwidth]{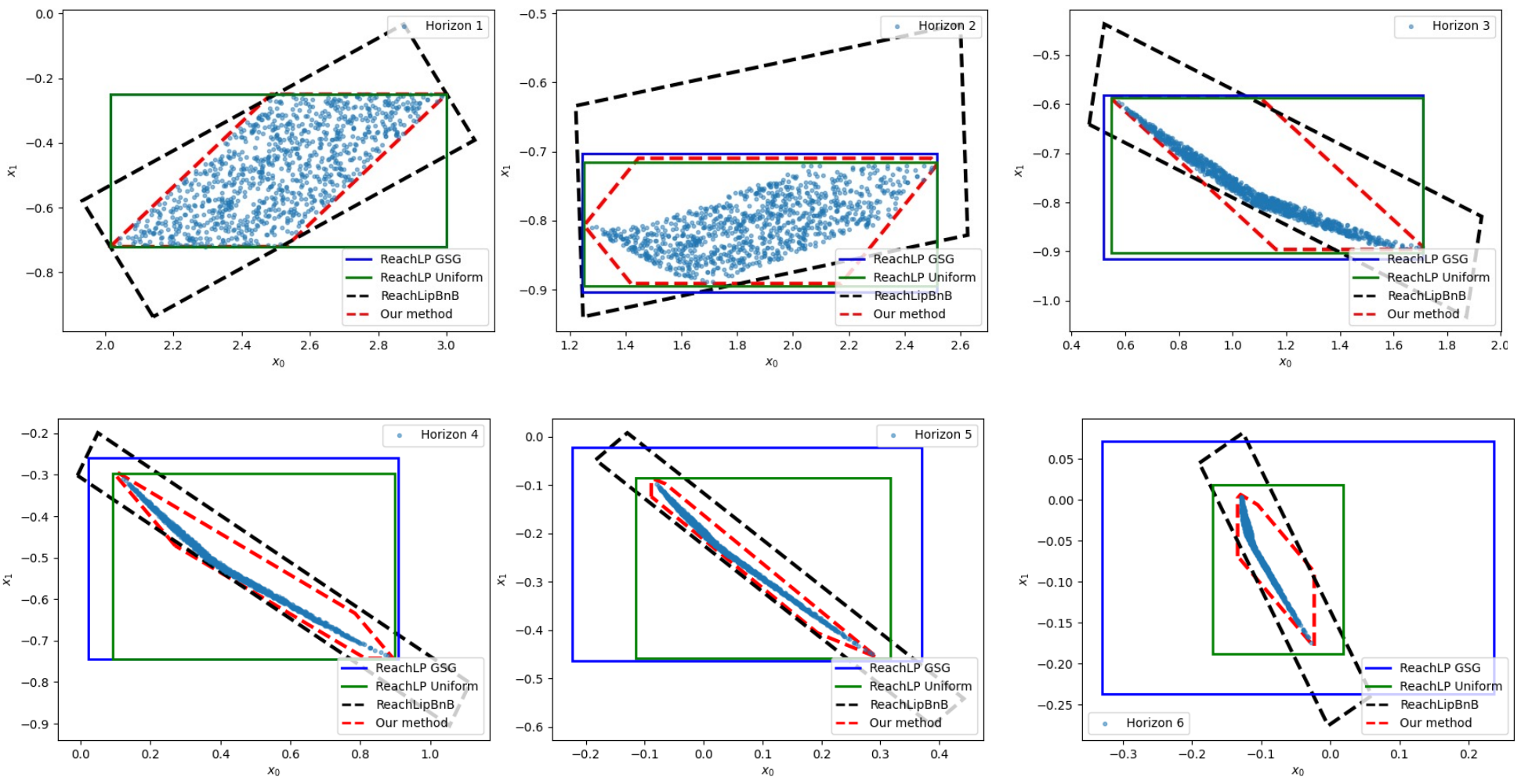}
    \caption{Computation of reachable sets for the double integrator system.}
    \label{fig:doubleIntegratorReach}
\end{figure}

\section{Experiments}
In this section, we present the results of our algorithm on two tasks and compare our results with the BnB method of \cite{entesari2022reachlipbnb} and the ReachLP method of \cite{everett2021reachability}.
The experiments are conducted on an Intel Xeon W-2245 3.9 GHz processor with 64 GB of RAM. For specifications of the BnB algorithm, see \cite{entesari2022reachlipbnb}.
For each work, we use the template polytope directions that the original work has proposed. For \cite{everett2021reachability}, we use their CROWN propagator and test their results with both the uniform and the greedy simulation guidance (GSG) partitioners. We use a $\lambda = 0.98$ for the cosine similarity threshold.

\subsection{Double Integrator}
We first consider the two-dimensional discrete-time system from \cite{hu2020reach}, for which we have $A = 
    \begin{bmatrix}
    1 & 1\\
    0 & 1
    \end{bmatrix}$,
    $B = 
    \begin{bmatrix}
    0.5\\
    1
    \end{bmatrix}$, and $e = 0$ in \eqref{eqn:linearSystem}. 
    We use a trained neural network that approximates the MPC controller. The neural network has the following number of neurons in each layer: (2-10-5-1) (starting from the input layer).
    The comparison is presented in Fig. \ref{fig:doubleIntegratorReach}. For this example, we use the initial set given by $\mathcal{X}^0 = \lbrack 2.5, 3\rbrack \times \lbrack -0.25, 0.25 \rbrack$. 
    As the figure suggests, our method provides considerably less conservative over-approximations, especially for longer time horizons. 
    Entry DI in Table  \ref{tab:runTimeResults} represents the results of the experiment. For this system, with our trained network, our method solves a total of 74 optimization problems, whereas the other methods solve 4 problems per time step. We used an epsilon solve accuracy of $0.01$ for ReachLipBnB.
    % For comparison with \cite{everett2021reachability}, we use both of their partitioning algorithms, i.e., greedy simulation guidance (GSG) and uniform, with the CROWN propagator. The run times of the algorithms are presented in table \ref{tab:runTimeResults}.

% \begin{table}[h!]
% \setlength{\tabcolsep}{2pt}
% \centering
%     \begin{tabular}{|c | c | c | c| c|} 
%      \hline
%      & ReachLipBnB  &  ReachLP GSG & ReachLP Uniform & Our method \\ [0.5ex] 
%      \hline
%      Run time[s] & 1.8 & 0.82 & 4.85 & 3.07 \\ 
%      \hline
%      Average run time/optimization [s] & 0.09 & 0.041 & 0.2425 & 0.038 \\
%      \hline
%     \end{tabular}
% \caption{\small Average run time for 12 time steps of the double integrator system for different methods. ReachLipBnB is run with an epsilon of 0.01. ReachLipBnB and ReachLP solve 4 optimization problems in for each time step. Our method solves an overall number of 80 optimization problems for this setting of 5 time steps.}
% \label{tab:doubleIntegrator}
% \end{table}

\subsection{6D Quadrotor}
We consider a quadrotor example with 6 state variables. This system is also taken from \cite{hu2020reach}. For this system, we have
$A = I_{6\times6} + \!\Delta t \times \begin{bmatrix}
0_{3\times3} & I_{3\times3} \\
0_{3\times3} & 0_{3\times3}
\end{bmatrix}$, 
$B = \!\Delta t \!\times\begin{bmatrix}
& g & 0 & 0 \\
0_{3\times3} & 0 & -g & 0 \\
& 0 & 0 & 1
\end{bmatrix}^\top$, 
$u = \begin{bmatrix}
\tan(\theta) &
\tan(\phi)&
\tau
\end{bmatrix}^\top$ and 
$e^\top = \!\Delta t \!\times\begin{bmatrix}
0_{5\times1} &
-g^\top
\end{bmatrix}$ in \eqref{eqn:linearSystem}.
We train a neural network to approximate an MPC controller. The neural network has 3 linear layers with 32, 32, and 3 neurons, respectively. As the original system is continuous time, we have discretized it with a sampling time of $\!\Delta t = 0.1$ seconds. The reachability analysis is conducted for 12 time steps.
%which is equivalent to 1.2 seconds in the original system.
For this example, we let the initial set be defined by $ [4.69, 4.71]\times [4.65, 4.75] \times [2.975, 3.025]\times [0.9499, 0.9501]\times [-0.0001, 0.0001]\times [-0.0001, 0.0001]$. %For this system, the uniform partitioning of \cite{everett2021reachability} did not return any results when run for over a day and we timed out. 
For this example, Fig. \ref{fig:quadReach} shows that even after 12 time steps, our method provides tight reachable sets.
Entry QR in Table  \ref{tab:runTimeResults} represents the results of the experiment. Our method solves a total of 426 optimization problems, whereas ReachLP and ReachLipBnB solve 12 and 16 problems per time step, respectively. The uniform partitioner for ReachLP timed out after 27 hours. We used an epsilon solve accuracy of $0.001$ for ReachLipBnB.

% \begin{figure}[t!]
%     \centering
%     \includegraphics[width=.6\textwidth]{images/quadRotorDeepPoly}
%     \caption{Reachable sets of the first two state variables of the quadrotor system. }
%     \label{fig:quadReach}
% \end{figure}

% \begin{table}[h!]
% \setlength{\tabcolsep}{2pt}
% \centering
%     \begin{tabular}{|c | c | c| c|} 
%      \hline
%      & ReachLipBnB  &  ReachLP & Our method \\ [0.5ex] 
%      \hline
%      Run time[s] &  2170  & 18.8 & 278 \\ 
%      \hline
%      Average run time/optimization & 11.3 & 0.13 & 0.49\\
%      \hline
%     \end{tabular}
% \caption{\small Time comparison for 12 time steps of the quadrotor system for different methods. The Reach-LP algorithm is run with GSG as the partitioner and CROWN as the propagator. ReachLipBnB is run with an epsilon of 0.001. ReachLipBnB solves 16 directions in each step, ReachLP solves 12 directions in each step, and our method solves a variable number of directions, typically increasingly as the time steps go on. For this system and network, our method solves an overall number of 570 optimization problems.}
% \label{tab:quadRotor}
% \end{table}
\begin{figure}[t!]
    \begin{minipage}{.5\textwidth}
    \includegraphics[width=0.9\textwidth]{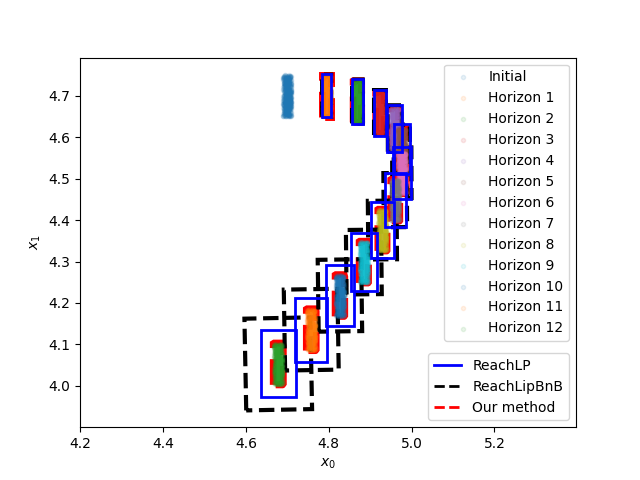}
    \captionof{figure}{Reachable sets of the first two state variables of the quadrotor system. The reachable sets of our method are tight in this experimental setup.}
     \label{fig:quadReach}
\end{minipage}
\begin{minipage}{.46\textwidth}
    
\centering
\resizebox{\textwidth}{!}{
\begin{tabular}{|l|l|c|c|c|c|}
\hline                                 &                       & \multicolumn{1}{l}{ReachLipBnB} & \multicolumn{1}{l}{ReachLP GSG} & \multicolumn{1}{l}{ReachLP uniform} & \multicolumn{1}{l}{Our method}\vline\\
\hline
\multirow{2}{*}{DI} & Run time [s]          & 2.45                             & 1.078                            & 5.95                                & 9.96                            \\
                                             & RT/ND & 0.102                            & 0.045                           & 0.248                              & 0.134                           \\
\hline
\multirow{2}{*}{\rotcell{QR}}         & Runt time [s]         & 2170                            & 18.8                            & -                                   & 213.67                             \\
                                             & RT/ND & 11.3                            & 0.13                            & -                                   & 0.501                          \\
                                             \hline
\end{tabular}
}
\captionof{table}{\small Run time statistics for experiments. RT/ND is the total run time of the experiment divided by the total number of optimization problems (or directions) solved. 
% a) For the double integrator (DI), we clock the total run time for 5 time steps. ReachLipBnB and ReachLP solve 4 optimization problems in each time step, whereas our method solves an overall number of 60 optimization problems in 5 time steps. ReachLipBnB is run within $0.01$ solve accuracy.
% b) For the quadrotor (QR), we run the experiment for 12 time steps. ReachLP and ReachLipBnB solve 12 and 16 problems in each time step, respectively, whereas our method solves a total of 426. ReachLipBnB is run within $0.001$ solve accuracy. We timed out ReachLP with uniform partitioner after 27 hours.
}

\label{tab:runTimeResults}
\end{minipage}
\end{figure}
\section{Conclusion}
We presented a novel method for reachability analysis of affine systems in feedback with $\mathrm{ReLU}$ neural network controllers using dynamic template polytopes. We then computed tight polyhedral over-approximation of the reachable sets using a BnB method based on partitioning the set of initial states and unrolling the dynamics. 
The bounding mechanism of DeepPoly provides fast estimates that scale up to large and deep neural networks; however, the application of our framework for higher dimensional input sets might be restrictive due to the worst-case exponential growth in partitioning the input set. To counteract this, we will explore less conservative bound propagation methods such as \cite{fatnassi2023bern},
%A potential disadvantage of our method is that partitioning can become less efficient for higher state space dimensions. To mitigate this, we will explore tighter bound propagation methods
combined with branching in the activation space (as opposed to the state space). Furthermore, we will explore operator splitting methods similar to \cite{chen2022deepsplit} to further improve scalability with respect to long time horizons.

\clearpage
\bibliography{bib.bib}

\begin{thebibliography}{38}
\providecommand{\natexlab}[1]{#1}
\providecommand{\url}[1]{\texttt{#1}}
\expandafter\ifx\csname urlstyle\endcsname\relax
  \providecommand{\doi}[1]{doi: #1}\else
  \providecommand{\doi}{doi: \begingroup \urlstyle{rm}\Url}\fi

\bibitem[Althoff et~al.(2021)Althoff, Frehse, and Girard]{althoff2021set}
Matthias Althoff, Goran Frehse, and Antoine Girard.
\newblock Set propagation techniques for reachability analysis.
\newblock \emph{Annual Review of Control, Robotics, and Autonomous Systems},
  4\penalty0 (1), 2021.

\bibitem[Ben~Sassi et~al.(2012)Ben~Sassi, Testylier, Dang, and
  Girard]{ben2012reachability}
Mohamed~Amin Ben~Sassi, Romain Testylier, Thao Dang, and Antoine Girard.
\newblock Reachability analysis of polynomial systems using linear programming
  relaxations.
\newblock In \emph{International Symposium on Automated Technology for
  Verification and Analysis}, pages 137--151. Springer, 2012.

\bibitem[Bogomolov et~al.(2017)Bogomolov, Frehse, Giacobbe, and
  Henzinger]{bogomolov2017counterexample}
Sergiy Bogomolov, Goran Frehse, Mirco Giacobbe, and Thomas~A Henzinger.
\newblock Counterexample-guided refinement of template polyhedra.
\newblock In \emph{International Conference on Tools and Algorithms for the
  Construction and Analysis of Systems}, pages 589--606. Springer, 2017.

\bibitem[Boyd and Mattingley(2007)]{boyd2007branch}
Stephen Boyd and Jacob Mattingley.
\newblock Branch and bound methods.
\newblock \emph{Notes for EE364b, Stanford University}, pages 2006--07, 2007.

\bibitem[Bradley and Manna(2007)]{bradley2007calculus}
Aaron~R Bradley and Zohar Manna.
\newblock \emph{The calculus of computation: decision procedures with
  applications to verification}.
\newblock Springer Science \& Business Media, 2007.

\bibitem[Bunel et~al.(2020)Bunel, De~Palma, Desmaison, Dvijotham, Kohli, Torr,
  and Kumar]{bunel2020lagrangian}
Rudy Bunel, Alessandro De~Palma, Alban Desmaison, Krishnamurthy Dvijotham,
  Pushmeet Kohli, Philip Torr, and M~Pawan Kumar.
\newblock Lagrangian decomposition for neural network verification.
\newblock In \emph{Conference on Uncertainty in Artificial Intelligence}, pages
  370--379. PMLR, 2020.

\bibitem[Chen et~al.(2022{\natexlab{a}})Chen, Preciado, and
  Fazlyab]{chen2022one}
Shaoru Chen, Victor~M Preciado, and Mahyar Fazlyab.
\newblock One-shot reachability analysis of neural network dynamical systems.
\newblock \emph{arXiv preprint arXiv:2209.11827}, 2022{\natexlab{a}}.

\bibitem[Chen et~al.(2022{\natexlab{b}})Chen, Wong, Kolter, and
  Fazlyab]{chen2022deepsplit}
Shaoru Chen, Eric Wong, J~Zico Kolter, and Mahyar Fazlyab.
\newblock Deepsplit: Scalable verification of deep neural networks via operator
  splitting.
\newblock \emph{IEEE Open Journal of Control Systems}, 1:\penalty0 126--140,
  2022{\natexlab{b}}.

\bibitem[Cheng et~al.(2017)Cheng, N{\"u}hrenberg, and Ruess]{cheng2017maximum}
Chih-Hong Cheng, Georg N{\"u}hrenberg, and Harald Ruess.
\newblock Maximum resilience of artificial neural networks.
\newblock In \emph{International Symposium on Automated Technology for
  Verification and Analysis}, pages 251--268. Springer, 2017.

\bibitem[Clavi{\`e}re et~al.(2021)Clavi{\`e}re, Asselin, Garion, and
  Pagetti]{claviere2021safety}
Arthur Clavi{\`e}re, Eric Asselin, Christophe Garion, and Claire Pagetti.
\newblock Safety verification of neural network controlled systems.
\newblock In \emph{2021 51st Annual IEEE/IFIP International Conference on
  Dependable Systems and Networks Workshops (DSN-W)}, pages 47--54. IEEE, 2021.

\bibitem[De~Palma et~al.(2021)De~Palma, Bunel, Desmaison, Dvijotham, Kohli,
  Torr, and Kumar]{de2021improved}
Alessandro De~Palma, Rudy Bunel, Alban Desmaison, Krishnamurthy Dvijotham,
  Pushmeet Kohli, Philip~HS Torr, and M~Pawan Kumar.
\newblock Improved branch and bound for neural network verification via
  lagrangian decomposition.
\newblock \emph{arXiv preprint arXiv:2104.06718}, 2021.

\bibitem[Dutta et~al.(2018{\natexlab{a}})Dutta, Jha, Sankaranarayanan, and
  Tiwari]{dutta2018learning}
Souradeep Dutta, Susmit Jha, Sriram Sankaranarayanan, and Ashish Tiwari.
\newblock Learning and verification of feedback control systems using
  feedforward neural networks.
\newblock \emph{IFAC-PapersOnLine}, 51\penalty0 (16):\penalty0 151--156,
  2018{\natexlab{a}}.

\bibitem[Dutta et~al.(2018{\natexlab{b}})Dutta, Jha, Sankaranarayanan, and
  Tiwari]{dutta2018output}
Souradeep Dutta, Susmit Jha, Sriram Sankaranarayanan, and Ashish Tiwari.
\newblock Output range analysis for deep feedforward neural networks.
\newblock In \emph{NASA Formal Methods Symposium}, pages 121--138. Springer,
  2018{\natexlab{b}}.

\bibitem[Dutta et~al.(2019)Dutta, Chen, and
  Sankaranarayanan]{dutta2019reachability}
Souradeep Dutta, Xin Chen, and Sriram Sankaranarayanan.
\newblock Reachability analysis for neural feedback systems using regressive
  polynomial rule inference.
\newblock In \emph{Proceedings of the 22nd ACM International Conference on
  Hybrid Systems: Computation and Control}, pages 157--168, 2019.

\bibitem[Entesari et~al.(2022)Entesari, Sharifi, and
  Fazlyab]{entesari2022reachlipbnb}
Taha Entesari, Sina Sharifi, and Mahyar Fazlyab.
\newblock Reachlipbnb: A branch-and-bound method for reachability analysis of
  neural autonomous systems using lipschitz bounds.
\newblock \emph{arXiv preprint arXiv:2211.00608}, 2022.

\bibitem[Everett et~al.(2021)Everett, Habibi, Sun, and
  How]{everett2021reachability}
Michael Everett, Golnaz Habibi, Chuangchuang Sun, and Jonathan~P How.
\newblock Reachability analysis of neural feedback loops.
\newblock \emph{IEEE Access}, 9:\penalty0 163938--163953, 2021.

\bibitem[Fatnassi et~al.(2023)Fatnassi, Khedr, Yamamoto, and
  Shoukry]{fatnassi2023bern}
Wael Fatnassi, Haitham Khedr, Valen Yamamoto, and Yasser Shoukry.
\newblock Bern-nn: Tight bound propagation for neural networks using bernstein
  polynomial interval arithmetic.
\newblock In \emph{Proceedings of the 26th ACM International Conference on
  Hybrid Systems: Computation and Control}, pages 1--11, 2023.

\bibitem[Ferrari et~al.(2022)Ferrari, Muller, Jovanovic, and
  Vechev]{ferrari2022complete}
Claudio Ferrari, Mark~Niklas Muller, Nikola Jovanovic, and Martin Vechev.
\newblock Complete verification via multi-neuron relaxation guided
  branch-and-bound.
\newblock \emph{arXiv preprint arXiv:2205.00263}, 2022.

\bibitem[Fischetti and Jo(2018)]{fischetti2018deep}
Matteo Fischetti and Jason Jo.
\newblock Deep neural networks and mixed integer linear optimization.
\newblock \emph{Constraints}, 23\penalty0 (3):\penalty0 296--309, 2018.

\bibitem[Hu et~al.(2020)Hu, Fazlyab, Morari, and Pappas]{hu2020reach}
Haimin Hu, Mahyar Fazlyab, Manfred Morari, and George~J Pappas.
\newblock Reach-sdp: Reachability analysis of closed-loop systems with neural
  network controllers via semidefinite programming.
\newblock In \emph{2020 59th IEEE Conference on Decision and Control (CDC)},
  pages 5929--5934. IEEE, 2020.

\bibitem[Huang et~al.(2019)Huang, Fan, Li, Chen, and Zhu]{huang2019reachnn}
Chao Huang, Jiameng Fan, Wenchao Li, Xin Chen, and Qi~Zhu.
\newblock Reachnn: Reachability analysis of neural-network controlled systems.
\newblock \emph{ACM Transactions on Embedded Computing Systems (TECS)},
  18\penalty0 (5s):\penalty0 1--22, 2019.

\bibitem[Ivanov et~al.(2021)Ivanov, Carpenter, Weimer, Alur, Pappas, and
  Lee]{ivanov2021verisig}
Radoslav Ivanov, Taylor Carpenter, James Weimer, Rajeev Alur, George Pappas,
  and Insup Lee.
\newblock Verisig 2.0: Verification of neural network controllers using taylor
  model preconditioning.
\newblock In \emph{International Conference on Computer Aided Verification},
  pages 249--262. Springer, 2021.

\bibitem[Kochdumper et~al.(2022)Kochdumper, Schilling, Althoff, and
  Bak]{kochdumper2022open}
Niklas Kochdumper, Christian Schilling, Matthias Althoff, and Stanley Bak.
\newblock Open-and closed-loop neural network verification using polynomial
  zonotopes.
\newblock \emph{arXiv preprint arXiv:2207.02715}, 2022.

\bibitem[Kouvaros and Lomuscio(2021)]{kouvaros2021towards}
Panagiotis Kouvaros and Alessio Lomuscio.
\newblock Towards scalable complete verification of relu neural networks via
  dependency-based branching.
\newblock In \emph{IJCAI}, pages 2643--2650, 2021.

\bibitem[Lomuscio and Maganti(2017)]{lomuscio2017approach}
Alessio Lomuscio and Lalit Maganti.
\newblock An approach to reachability analysis for feed-forward relu neural
  networks.
\newblock \emph{arXiv preprint arXiv:1706.07351}, 2017.

\bibitem[Neumaier(1993)]{neumaier1993wrapping}
Arnold Neumaier.
\newblock The wrapping effect, ellipsoid arithmetic, stability and confidence
  regions.
\newblock In \emph{Validation numerics}, pages 175--190. Springer, 1993.

\bibitem[Rober et~al.(2022)Rober, Katz, Sidrane, Yel, Everett, Kochenderfer,
  and How]{rober2022backward}
Nicholas Rober, Sydney~M Katz, Chelsea Sidrane, Esen Yel, Michael Everett,
  Mykel~J Kochenderfer, and Jonathan~P How.
\newblock Backward reachability analysis of neural feedback loops: Techniques
  for linear and nonlinear systems.
\newblock \emph{arXiv preprint arXiv:2209.14076}, 2022.

\bibitem[Schilling et~al.(2022)Schilling, Forets, and
  Guadalupe]{schilling2022verification}
Christian Schilling, Marcelo Forets, and Sebasti{\'a}n Guadalupe.
\newblock Verification of neural-network control systems by integrating taylor
  models and zonotopes.
\newblock In \emph{Proceedings of the AAAI Conference on Artificial
  Intelligence}, volume~36, pages 8169--8177, 2022.

\bibitem[Sidrane et~al.(2022)Sidrane, Maleki, Irfan, and
  Kochenderfer]{sidrane2022overt}
Chelsea Sidrane, Amir Maleki, Ahmed Irfan, and Mykel~J Kochenderfer.
\newblock Overt: An algorithm for safety verification of neural network control
  policies for nonlinear systems.
\newblock \emph{Journal of Machine Learning Research}, 23\penalty0
  (117):\penalty0 1--45, 2022.

\bibitem[Singh et~al.(2019)Singh, Gehr, P{\"u}schel, and
  Vechev]{singh2019abstract}
Gagandeep Singh, Timon Gehr, Markus P{\"u}schel, and Martin Vechev.
\newblock An abstract domain for certifying neural networks.
\newblock \emph{Proceedings of the ACM on Programming Languages}, 3\penalty0
  (POPL):\penalty0 1--30, 2019.

\bibitem[Stursberg and Krogh(2003)]{stursberg2003efficient}
Olaf Stursberg and Bruce~H Krogh.
\newblock Efficient representation and computation of reachable sets for hybrid
  systems.
\newblock In \emph{International Workshop on Hybrid Systems: Computation and
  Control}, pages 482--497. Springer, 2003.

\bibitem[Tjeng et~al.(2017)Tjeng, Xiao, and Tedrake]{tjeng2017evaluating}
Vincent Tjeng, Kai Xiao, and Russ Tedrake.
\newblock Evaluating robustness of neural networks with mixed integer
  programming.
\newblock \emph{arXiv preprint arXiv:1711.07356}, 2017.

\bibitem[Tran et~al.(2019)Tran, Manzanas~Lopez, Musau, Yang, Nguyen, Xiang, and
  Johnson]{tran2019star}
Hoang-Dung Tran, Diago Manzanas~Lopez, Patrick Musau, Xiaodong Yang, Luan~Viet
  Nguyen, Weiming Xiang, and Taylor~T Johnson.
\newblock Star-based reachability analysis of deep neural networks.
\newblock In \emph{International symposium on formal methods}, pages 670--686.
  Springer, 2019.

\bibitem[Vincent and Schwager(2021)]{vincent2021reachable}
Joseph~A Vincent and Mac Schwager.
\newblock Reachable polyhedral marching (rpm): A safety verification algorithm
  for robotic systems with deep neural network components.
\newblock In \emph{2021 IEEE International Conference on Robotics and
  Automation (ICRA)}, pages 9029--9035. IEEE, 2021.

\bibitem[Wang et~al.(2021)Wang, Zhang, Xu, Lin, Jana, Hsieh, and
  Kolter]{wang2021beta}
Shiqi Wang, Huan Zhang, Kaidi Xu, Xue Lin, Suman Jana, Cho-Jui Hsieh, and
  J~Zico Kolter.
\newblock Beta-crown: Efficient bound propagation with per-neuron split
  constraints for neural network robustness verification.
\newblock \emph{Advances in Neural Information Processing Systems},
  34:\penalty0 29909--29921, 2021.

\bibitem[Xu et~al.(2020)Xu, Zhang, Wang, Wang, Jana, Lin, and
  Hsieh]{xu2020fast}
Kaidi Xu, Huan Zhang, Shiqi Wang, Yihan Wang, Suman Jana, Xue Lin, and Cho-Jui
  Hsieh.
\newblock Fast and complete: Enabling complete neural network verification with
  rapid and massively parallel incomplete verifiers.
\newblock \emph{arXiv preprint arXiv:2011.13824}, 2020.

\bibitem[Zhang et~al.(2018)Zhang, Weng, Chen, Hsieh, and
  Daniel]{zhang2018efficient}
Huan Zhang, Tsui-Wei Weng, Pin-Yu Chen, Cho-Jui Hsieh, and Luca Daniel.
\newblock Efficient neural network robustness certification with general
  activation functions.
\newblock \emph{Advances in neural information processing systems}, 31, 2018.

\bibitem[Zhang and Xu(2022)]{zhang2022reachability}
Yuhao Zhang and Xiangru Xu.
\newblock Reachability analysis and safety verification of neural feedback
  systems via hybrid zonotopes.
\newblock \emph{arXiv preprint arXiv:2210.03244}, 2022.

\end{thebibliography}

\end{document}